\documentclass[%
reprint,
amsmath,amssymb,
aps,twocolumn,superscriptaddress
]{revtex4-2}

\usepackage{graphicx}
\usepackage{amsmath}
\usepackage{dcolumn}
\usepackage{bm}
\usepackage{bbm}
\usepackage{float}
\usepackage{braket}
\usepackage{tcolorbox}

\usepackage{amssymb,amsthm}

\theoremstyle{definition}
\newtheorem{definition}{Definition}

\theoremstyle{plain}
\newtheorem{proposition}{Proposition}

\usepackage[colorlinks=true,linkcolor=blue,citecolor=blue,urlcolor=blue]{hyperref}

\begin{document}
\title{Symmetry-resolved properties of the trace distance in thermalizing SU(2) systems}
\author{Haojie Shen}
\email{haojieshen@nju.edu.cn}
\affiliation{National Laboratory of Solid State Microstructures and School of Physics, Nanjing University, Nanjing 210093, China}
\affiliation{Collaborative Innovation Center of Advanced Microstructures, Nanjing University, Nanjing 210093, China}
\affiliation{School of Physics and Astronomy, Shanghai Jiao Tong University, Shanghai 200240, China}

\author{Jie Chen}
\email{chenjie666@xhu.edu.cn}
\affiliation{School of Science, Key Laboratory of High Performance Scientific Computation, Xihua University, Chengdu 610039, China}

\author{Xiaoqun Wang}
\email{xiaoqunwang@zju.edu.cn}
\affiliation{School of Physics, Zhejiang University, Hangzhou 310027, China}

\date{\today}

\begin{abstract}
We study diagnostics of thermalization in quantum many-body systems with global SU(2) symmetry, where the standard eigenstate thermalization hypothesis (ETH) is generalized to its non-Abelian form. As an eigenstate-level probe, we introduce a symmetry-resolved trace distance constructed from the block structure of the reduced density matrix. This block structure separates spin-sector probabilities from configurational fluctuations within each sector, naturally leading to a decomposition into a probability trace distance and a configurational trace distance. The microcanonical average of the former is bounded by fluctuations of the corresponding spin-sector probabilities within a microcanonical energy window, whereas the latter captures finer intra-sector fluctuations. In non-Abelian thermalizing systems, these spin-sector-probability fluctuations are constrained by the non-Abelian ETH and therefore become exponentially suppressed with system size. Numerical studies of the one-dimensional \(J_1\)--\(J_2\) Heisenberg chain are consistent with this picture and suggest that, in the thermal regime, the trace distance is asymptotically dominated by the configurational trace distance.
\end{abstract}

\maketitle

\textit{Introduction.}---Thermalization in isolated quantum many-body systems is commonly understood through the eigenstate thermalization hypothesis (ETH)~\cite{Deutsch1991,Srednicki1994,Rigol2008}. For systems with noncommuting conserved charges, however, thermalization is more subtle: the corresponding equilibrium ensembles are non-Abelian, and noncommutativity has been shown to modify thermal-state structure~\cite{Halpern2016,Halpern2018,Halpern2020,Majidy2023}, increase entanglement entropy~\cite{MajidyEntanglement2023}, and even be probed in trapped-ion experiments where small subsystems equilibrate close to a predicted non-Abelian thermal state~\cite{Kranzl2023}. At the eigenstate level, the corresponding structure is encoded in the non-Abelian ETH~\cite{Noh2023,Murthy2023,Lasek2024,Patil2025}. Using the non-Abelian ETH, Ref.~\cite{NohKMS2025} further derived a Kubo--Martin--Schwinger relation for energy eigenstates of SU(2)-symmetric quantum many-body systems. These developments naturally motivate asking how the non-Abelian ETH is reflected in diagnostics of thermalization?

In SU(2)-symmetric systems, the Wigner--Eckart theorem separates matrix elements of an irreducible tensor operator \(T^{(k)}\) into a symmetry part fixed by Clebsch--Gordan coefficients and a dynamical part carried by reduced matrix elements. Since our analysis builds on this structure, we briefly recall the non-Abelian ETH. Denoting \(\bm{\alpha} = (E_\alpha, S_\alpha)\) as the state label, the non-Abelian ETH takes the form
\begin{equation}
\begin{aligned}
\Braket{\bm{\alpha} \| T^{(k)} \| \bm{\alpha}'} 
&= \mathcal{T}^{(k)}(\mathcal{E}, \mathcal{S}) \, \delta_{\bm{\alpha}\bm{\alpha}'} \\
&+ e^{-S_{\mathrm{th}}(\mathcal{E}, \mathcal{S})/2} \, f^{(T)}_\nu(\mathcal{E}, \mathcal{S}, \omega) \, R^{(T)}_{\bm{\alpha}\bm{\alpha}'},\label{eq:NAETH}    
\end{aligned}
\end{equation}
where \(\mathcal{E} = (E_\alpha + E_{\alpha'})/2\), \(\omega = E_\alpha - E_{\alpha'}\), \(\mathcal{S} = (S_\alpha + S_{\alpha'})/2\), and \(\nu = S_\alpha - S_{\alpha'}\). Here, \(S_{\mathrm{th}}(\mathcal{E}, \mathcal{S})\) denotes the thermodynamic entropy in the subspace of given energy and spin, \(\mathcal{T}^{(k)}(\mathcal{E}, \mathcal{S})\) and \(f^{(T)}_\nu(\mathcal{E}, \mathcal{S}, \omega)\) are smooth functions, and \(R^{(T)}_{\bm{\alpha}\bm{\alpha}'}\) is a normalized random variable capturing off-diagonal fluctuations.

Equation~\eqref{eq:NAETH} describes the smoothness of neighboring energy eigenstates, up to exponentially small fluctuations. This suggests using the trace distance between their reduced density matrices as an eigenstate-based diagnostic of thermalization, following Ref.~\cite{Khasseh2023}.
In SU(2)-symmetric systems, the reduced density matrix is block diagonal in sectors labeled by \(S_A\), the total-spin quantum number of subsystem \(A\). This block structure, much as in symmetry-resolved entanglement entropy~\cite{Goldstein2018,Monkman2023,Bianchi2024,Chen2024FrontPhys,Chen2025PRB}, induces a corresponding decomposition of the trace distance. One part, which we call the probability trace distance, is directly tied to the spin-sector probabilities; the other, the configurational trace distance, captures the remaining intra-sector fluctuations. 

In this work, we exploit this structure to diagnose thermalization from neighboring eigenstates in SU(2)-symmetric many-body systems. Our first statement decomposes and bounds the trace distance in terms of these two contributions. Our second statement shows that the microcanonical average of the probability trace distance is controlled by fluctuations of the spin-sector probabilities within a microcanonical energy window, and that in non-Abelian thermalizing systems the non-Abelian ETH describes the scaling of these fluctuations and thereby drives an exponential suppression with system size. We then test these predictions against exact diagonalization results for the one-dimensional \(J_1\)\(-\)\(J_2\) Heisenberg chain.

\textit{Setup.}---
We first define the projectors, reduced density matrices, and trace-distance diagnostics used in the rest of the paper.

Let $\ket{\psi_{SM}}$ be a many-body eigenstate belonging to a fixed global $SU(2)$ multiplet $(S,M)$, and let $A\cup B$ be a bipartition of the system~\cite{Bianchi2024}. When needed, we characterize the bipartition by the subsystem fraction \(x=N_A/N\), where \(N_A\) is the number of sites in subsystem \(A\). In a spin-coupling basis, the state can be expanded as
\begin{equation}
\begin{aligned}
\ket{\psi_{SM}}
&=
\sum_{S_A}\sqrt{P_{S_A}}
\sum_{\tilde S_A,\tilde S_B}
\chi^{(S_A)}_{\tilde S_A,\tilde S_B}\\
&\quad \times
\ket{(\tilde S_A)S_A}_A
\ket{S_A,(\tilde S_B);SM}_B ,
\end{aligned}
\label{eq:wf_su2}
\end{equation}
where $S_A$ denotes the total spin of subsystem $A$, while $\tilde S_A$ and $\tilde S_B$ label the internal branches of the spin-coupling tree on the two subsystems. The coefficients satisfy
\begin{equation}
\sum_{\tilde S_A,\tilde S_B}
\left|\chi^{(S_A)}_{\tilde S_A,\tilde S_B}\right|^2 = 1,
\qquad
\sum_{S_A} P_{S_A}=1.
\end{equation}

\begin{definition}[Subsystem-spin projector]
For a fixed bipartition $A\cup B$, the projector onto the subspace in which subsystem $A$ carries total spin $S_A$ is defined by
\begin{equation}
\begin{aligned}
\Pi_{S_A}
&=
\sum_{\tilde S_A,\tilde S_B,S,M}
\ket{(\tilde S_A)S_A}_A
\ket{S_A,(\tilde S_B);SM}_B\\
&\quad \times
\bra{(\tilde S_A)S_A}_A
\bra{S_A,(\tilde S_B);SM}_B .
\end{aligned}
\label{eq:PiSA_def}
\end{equation}
The probability of finding subsystem \(A\) in the spin-\(S_A\) sector is
\begin{equation}
P_{S_A}
=
\bra{\psi_{SM}}\Pi_{S_A}\ket{\psi_{SM}}.
\label{eq:PSA_def}
\end{equation}
\end{definition}

\begin{proof}[Proof of $SU(2)$ invariance]
The projector $\Pi_{S_A}$ is invariant under global $SU(2)$ rotations. For any
\begin{equation}
U(R)=e^{-i\theta\,\hat n\cdot \mathbf S_{\mathrm{tot}}}, \qquad R\in SU(2),
\end{equation}
its action in the coupled basis is
\begin{equation}
\begin{aligned}
&U(R)\ket{(\tilde S_A)S_A}_A \ket{S_A,(\tilde S_B);SM}_B \\
&= \sum_{M'=-S}^{S} D^{(S)}_{M',M}(R)\, \ket{(\tilde S_A)S_A}_A \\
&\times \ket{S_A,(\tilde S_B);SM'}_B,
\end{aligned}
\label{eq:UR_action}
\end{equation}
where $D^{(S)}_{M',M}(R)$ is the Wigner $D$ matrix. Substituting Eq.~\eqref{eq:UR_action} into Eq.~\eqref{eq:PiSA_def} and using the unitarity relation
\begin{equation}
\sum_M D^{(S)}_{M',M}(R) D^{(S)\,*}_{M'',M}(R)=\delta_{M',M''},
\end{equation}
one immediately obtains
\begin{equation}
U(R)\Pi_{S_A}U(R)^\dagger=\Pi_{S_A}.
\end{equation}
Therefore,
\begin{equation}
[U(R),\Pi_{S_A}]=0, \qquad \forall R\in SU(2),
\end{equation}
or equivalently $[\mathbf S_{\mathrm{tot}},\Pi_{S_A}]=0$. In practice, we perform the numerics in a fixed-$M$ block. Owing to the global $SU(2)$ symmetry, the symmetry-resolved quantities defined below are independent of the particular choice of $M$.
\end{proof}

\noindent\textit{Note.} Although $\Pi_{S_A}$ is defined through the subsystem spin label $S_A$, it is a global projector written in the coupled basis of the full system and is generally not a few-body local operator.

\begin{definition}[Symmetry-resolved reduced density matrix]
Let
\begin{equation}
\rho = \ket{\psi_{SM}}\bra{\psi_{SM}}
\end{equation}
be the density matrix of the full system. The unnormalized reduced density matrix of subsystem $A$ in the spin-$S_A$ sector is defined as
\begin{equation}
\rho_A^{(S_A)}
:=
\mathrm{Tr}_B\!\left[\Pi_{S_A}\,\rho\,\Pi_{S_A}\right].
\label{eq:rhoA_SA_def}
\end{equation}
Its trace is the corresponding spin-sector probability,
\begin{equation}
\mathrm{Tr}_A\!\left[\rho_A^{(S_A)}\right]=P_{S_A}.
\label{eq:rhoA_SA_trace}
\end{equation}
Whenever $P_{S_A}>0$, the normalized reduced density matrix in this sector is defined by
\begin{equation}
\tilde\rho_A^{(S_A)}
:=
\frac{\rho_A^{(S_A)}}{P_{S_A}}.
\label{eq:rhoA_SA_normalized}
\end{equation}
The reduced density matrix of subsystem $A$ then admits the block decomposition
\begin{equation}
\rho_A
=
\mathrm{Tr}_B\,\rho
=
\bigoplus_{S_A}\rho_A^{(S_A)}
=
\bigoplus_{S_A} P_{S_A}\,\tilde\rho_A^{(S_A)}.
\label{eq:rhoA_directsum}
\end{equation}
\end{definition}

\begin{definition}[Symmetry-resolved trace distance]
For two neighboring eigenstates labeled by $\alpha$ and $\alpha+1$ within the same global symmetry sector, we define the trace distance between their reduced density matrices on subsystem $A$ as
\begin{equation}
D_\alpha^A
:=
\frac{1}{2}\big\|\rho_{A,\alpha+1}-\rho_{A,\alpha}\big\|_1.
\label{eq:DA_def}
\end{equation}
Because $\rho_A$ is block diagonal in Eq.~\eqref{eq:rhoA_directsum}, this quantity decomposes as
\begin{equation}
D_\alpha^A
=
\sum_{S_A} D_\alpha^{(S_A)},
\label{eq:DA_decomp}
\end{equation}
where
\begin{equation}
D_\alpha^{(S_A)}
:=
\frac{1}{2}\big\|\rho_{A,\alpha+1}^{(S_A)}-\rho_{A,\alpha}^{(S_A)}\big\|_1
\label{eq:D_SA_def}
\end{equation}
is the trace distance in the spin-$S_A$ sector. We refer to \(D_\alpha^A\) as the symmetry-resolved trace distance. Here $\|\cdot\|_1$ denotes the Schatten 1-norm.
\end{definition}

\textit{Statements.}---
With the notation in place, we now state the main results and give their proofs.

\begin{proposition}[Bounds for the trace distance]
For two neighboring eigenstates labeled by \(\alpha\) and \(\alpha+1\), the trace distance \(D_\alpha^A\) is bounded from below by the change in the spin-sector probabilities and from above by the sum of two contributions defined below.
\end{proposition}

\begin{proof}
For the lower bound, Hölder's inequality gives
\(
\big|\mathrm{Tr}(X)\big| \le \|X\|_1
\)
for any operator \(X\). Applying this to
\(
X=\rho_{A,\alpha+1}^{(S_A)}-\rho_{A,\alpha}^{(S_A)}
\),
we obtain
\begin{equation}
\begin{aligned}
\Delta P_{S_A}
&\equiv \big| P^{(\alpha+1)}_{S_A} - P^{(\alpha)}_{S_A} \big| \\
&= \big| \mathrm{Tr}\big[\rho_{A,\alpha+1}^{(S_A)} - \rho_{A,\alpha}^{(S_A)} \big] \big| \\
&\le \big\| \rho_{A,\alpha+1}^{(S_A)} - \rho_{A,\alpha}^{(S_A)} \big\|_1 \\
&= 2 D_\alpha^{(S_A)} .
\end{aligned}
\label{lowerbound}
\end{equation}

For the upper bound, applying the strong convexity of the trace distance~\cite{Nielsen_Chuang_2010} to the block decomposition, we obtain
\begin{equation}
\begin{aligned}
D_\alpha^{A}
&\le \tfrac{1}{2}\sum_{S_A} \big| P_{S_A}^{(\alpha+1)} - P_{S_A}^{(\alpha)} \big| \\
&\quad + \tfrac{1}{2}\sum_{S_A} P_{S_A}^{(\alpha+1)}
\big\| \tilde{\rho}_{A,\alpha+1}^{(S_A)} - \tilde{\rho}_{A,\alpha}^{(S_A)} \big\|_1\\
&\equiv D^{A}_{\alpha,\mathrm{prob}} + D^{A}_{\alpha,\mathrm{conf}}.
\end{aligned}\label{SRTD}
\end{equation}
Here
\(
D^{A}_{\alpha,\mathrm{prob}}
\)
denotes the contribution from the change in the spin-sector probabilities, while
\(
D^{A}_{\alpha,\mathrm{conf}}
\)
denotes the contribution from the change within each block of the reduced density matrix. We therefore refer to them as the probability trace distance and the configurational trace distance, respectively. Combining Eqs.~\eqref{lowerbound} and \eqref{SRTD}, we arrive at
\begin{equation}
   D^{A}_{\alpha,\mathrm{prob}} \le D_\alpha^{A} \le D^{A}_{\alpha,\mathrm{prob}} + D^{A}_{\alpha,\mathrm{conf}}.
\label{eq:D_bounds}
\end{equation}
\end{proof}

\begin{proposition}[Microcanonical average of \(D^{A}_{\alpha,\mathrm{prob}}\)]
Let \(\mathcal{W}\) be a narrow microcanonical window, namely a set of consecutive eigenstates within a narrow energy interval, and let
\(N_{\mathrm{sec}}\) denote the number of allowed spin sectors. Write
\(\overline{P_{S_A}} := \big\langle P_{S_A}^{(\alpha)} \big\rangle_{\mathcal{W}}\).
Define
\begin{equation}
\big\langle D^{A}_{\alpha,\mathrm{prob}} \big\rangle_{\mathcal{W}}
:=
\frac{1}{2}\sum_{S_A}
\big\langle
\big|P_{S_A}^{(\alpha+1)}-P_{S_A}^{(\alpha)}\big|
\big\rangle_{\mathcal{W}}.
\end{equation}
Then
\begin{equation}
\big\langle D^{A}_{\alpha,\mathrm{prob}} \big\rangle_{\mathcal{W}}
\le
\sqrt{
N_{\mathrm{sec}}
\left(
\sum_{S_A}\mathrm{Var}_{\mathcal{W}}\big(P_{S_A}^{(\alpha)}\big)
\right)
},
\label{eq:Dprob_micro_bound}
\end{equation}
For pair-dependent quantities, \(\langle\cdot\rangle_{\mathcal{W}}\) denotes the average over all \(\alpha\) such that both \(\alpha\) and \(\alpha+1\) belong to \(\mathcal{W}\). In particular, if
\begin{equation}
N_{\mathrm{sec}}
\sum_{S_A}\mathrm{Var}_{\mathcal{W}}\big(P_{S_A}^{(\alpha)}\big)
\to 0
\end{equation}
in the thermodynamic limit, then \(\langle D^{A}_{\alpha,\mathrm{prob}}\rangle_{\mathcal{W}}\to0\). When the non-Abelian ETH holds in \(\mathcal{W}\) and \(N_{\mathrm{sec}}\sum_{S_A}\mathrm{Var}_{\mathcal{W}}\big(\Delta_{S_A}^{(\bm{\alpha})}\big)\) grows at most polynomially with system size, this decay is exponential. The window is assumed to contain sufficiently many consecutive eigenstates so that endpoint corrections are negligible.
\end{proposition}

\begin{proof}
By the triangle inequality,
\begin{equation}
\big|P_{S_A}^{(\alpha+1)}-P_{S_A}^{(\alpha)}\big|
\le
\big|P_{S_A}^{(\alpha+1)}-\overline{P_{S_A}}\big|
+
\big|P_{S_A}^{(\alpha)}-\overline{P_{S_A}}\big|.
\end{equation}
Averaging over all admissible consecutive pairs in \(\mathcal{W}\), and relabeling \(\alpha+1\mapsto\alpha\) in the first averaged term, we obtain the following bound up to endpoint corrections of order \(O(|\mathcal{W}|^{-1})\), which are negligible for a sufficiently large window:
\begin{equation}
\big\langle D^{A}_{\alpha,\mathrm{prob}} \big\rangle_{\mathcal{W}}
\le
\sum_{S_A}
\big\langle
\big|P_{S_A}^{(\alpha)}-\overline{P_{S_A}}\big|
\big\rangle_{\mathcal{W}}.
\label{eq:Dprob_pair_bound}
\end{equation}
Using the root-mean-square--mean inequality,
\(
\langle |X| \rangle_{\mathcal{W}} \le \sqrt{\langle X^2 \rangle_{\mathcal{W}}}
\),
with \(X=P_{S_A}^{(\alpha)}-\overline{P_{S_A}}\), we arrive at
\begin{equation}
\big\langle D^{A}_{\alpha,\mathrm{prob}} \big\rangle_{\mathcal{W}}
\le
\sum_{S_A}\sqrt{\mathrm{Var}_{\mathcal{W}}\big(P_{S_A}^{(\alpha)}\big)}.
\end{equation}
Applying the Cauchy--Schwarz inequality to
\(
x_{S_A}=1
\)
and
\(
y_{S_A}=\sqrt{\mathrm{Var}_{\mathcal{W}}\big(P_{S_A}^{(\alpha)}\big)}
\),
and using
\(
\sum_{S_A} x_{S_A}^2 = N_{\mathrm{sec}}
\),
gives Eq.~\eqref{eq:Dprob_micro_bound}.

When the non-Abelian ETH holds in \(\mathcal{W}\), applying the diagonal part of Eq.~\eqref{eq:NAETH} to the \(SU(2)\)-invariant projector \(\Pi_{S_A}\) in Eq.~\eqref{eq:PSA_def} yields
\begin{equation}\label{eq:P_decomp}
P_{S_A}^{(\bm{\alpha})}
= \mathcal{T}^{(0)}_{S_A}(\mathcal{E},\mathcal{S})
+ e^{-S_{\mathrm{th}}(\mathcal{E},\mathcal{S})/2}\, \Delta_{S_A}^{(\bm{\alpha})},
\end{equation}
where \(\mathcal{T}^{(0)}_{S_A}\) is smooth within the narrow window and \(\Delta_{S_A}^{(\bm{\alpha})}\) denotes the state-to-state fluctuation. Therefore,
\begin{equation}\label{eq:VarDef}
\mathrm{Var}_{\mathcal{W}}\big(P_{S_A}^{(\bm{\alpha})}\big)
= e^{-S_{\mathrm{th}}(\mathcal{E},\mathcal{S})} \,
\mathrm{Var}_{\mathcal{W}}\big(\Delta_{S_A}^{(\bm{\alpha})}\big).
\end{equation}
Substituting Eq.~\eqref{eq:VarDef} into Eq.~\eqref{eq:Dprob_micro_bound}, we obtain
\begin{equation}
\big\langle D^{A}_{\alpha,\mathrm{prob}} \big\rangle_{\mathcal{W}}
\le
e^{-S_{\mathrm{th}}(\mathcal{E},\mathcal{S})/2}
\sqrt{
N_{\mathrm{sec}}
\sum_{S_A}\mathrm{Var}_{\mathcal{W}}\big(\Delta_{S_A}^{(\bm{\alpha})}\big)
}.
\label{D_prob}
\end{equation}
Since \(S_{\mathrm{th}}(\mathcal{E},\mathcal{S})\propto N\), the right-hand side vanishes exponentially with \(N\) provided \(N_{\mathrm{sec}}\sum_{S_A}\mathrm{Var}_{\mathcal{W}}\big(\Delta_{S_A}^{(\bm{\alpha})}\big)\) grows at most polynomially with system size.
\end{proof}

\begin{figure*}[!t]
    \centering
    \includegraphics[width=0.8\linewidth]{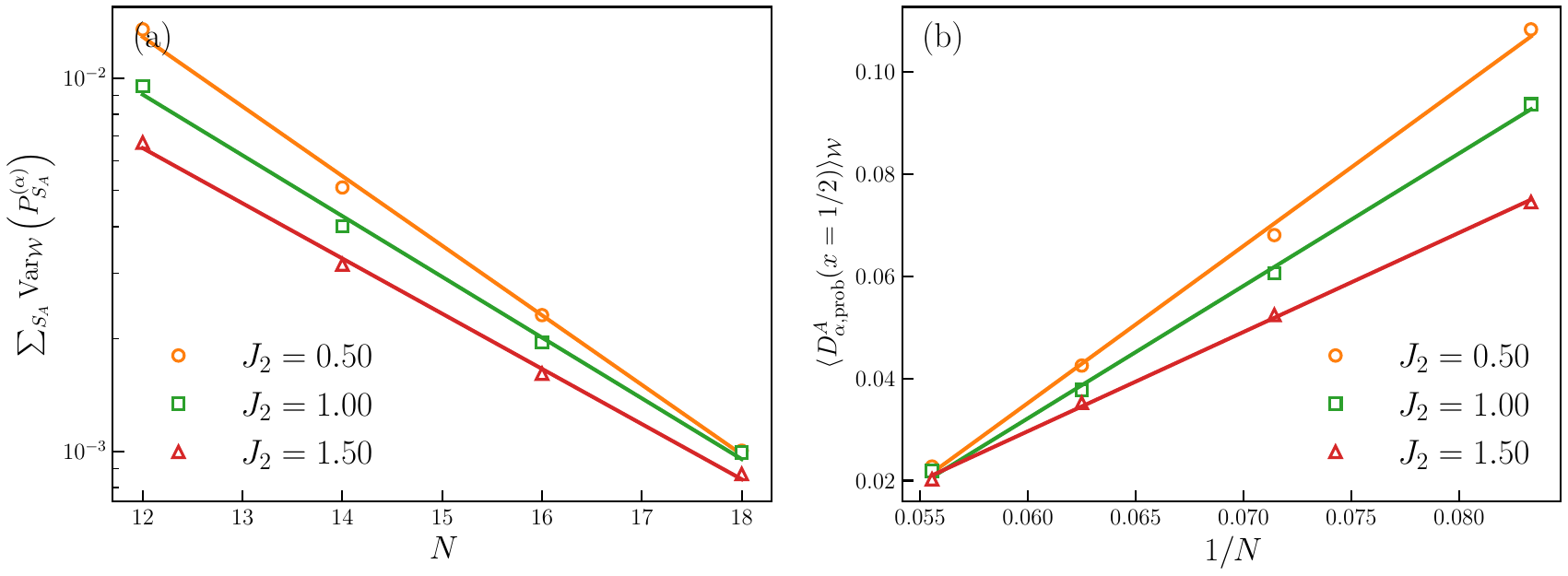}
   \caption{
(a) Finite-size scaling of the sum of variances of the spin-sector probabilities,
\(\sum_{S_A}\mathrm{Var}_{\mathcal{W}}\!\big(P_{S_A}^{(\bm{\alpha})}\big)\),
and
(b) the corresponding window average
\(\langle D^{A}_{\alpha,\mathrm{prob}}(x=1/2)\rangle_{\mathcal{W}}\).
Different colors correspond to different \(J_2\) values.
}
    \label{fig1}
\end{figure*}

\textit{Numerical results.}---
As a numerical illustration, we consider the one-dimensional \(J_1\)--\(J_2\) Heisenberg chain with open boundaries,
\begin{equation}
H = \sum_{j=1}^{N-1} J_1\,\mathbf S_j\!\cdot\!\mathbf S_{j+1}
  + J_2 \sum_{j=1}^{N-2} \mathbf S_j\!\cdot\!\mathbf S_{j+2},
\end{equation}
where \(J_1=1\). At \(J_2=0\) the model reduces to the integrable Heisenberg chain, while the nonzero \(J_2\) values considered below lie in the thermalizing regime. All data shown below are obtained by exact diagonalization (ED) in fixed global symmetry sectors for system sizes \(N=12,14,16,\) and \(18\). Within each such sector, eigenstates are ordered by energy, and the microcanonical window is chosen as the \(40\%\) to \(65\%\) portion of that ordered list to avoid edge effects~\cite{Beugeling2014}. Fig.~\ref{fig1} summarizes the finite-size scaling of the fluctuations of the spin-sector probabilities within the microcanonical energy window and of \(\langle D^{A}_{\alpha,\mathrm{prob}}(x=1/2)\rangle_{\mathcal{W}}\).

We first test Proposition 2, which bounds the microcanonical average of the probability trace distance in terms of the fluctuations of the spin-sector probabilities within the microcanonical energy window. Fig.~\ref{fig1}(a) shows the quantity
$
\sum_{S_A}  \mathrm{Var}_{\mathcal{W}}\big(P_{S_A}^{(\bm{\alpha})}\big),
$
which captures the scaling term entering the analytical bound. It decays exponentially with system size, in line with the expectation from the non-Abelian ETH. Fig.~\ref{fig1}(b) shows the corresponding window average \(\langle D^{A}_{\alpha,\mathrm{prob}}(x=1/2)\rangle_{\mathcal{W}}\), namely the quantity controlled by Proposition 2. Over the system sizes accessible numerically, this quantity is better described by a power law than by a clean exponential. We view this as a finite-size effect rather than a breakdown of Eq.~\eqref{D_prob}: the inequality in Eq.~\eqref{eq:Dprob_pair_bound} used to derive Eq.~\eqref{D_prob} is still far from saturated at these sizes. The observed power-law behavior should therefore be understood as a preasymptotic trend rather than a genuine contradiction with the thermalization prediction.

Turning back to Proposition 1, Fig.~\ref{fig2} tests whether the upper bound in Eq.~\eqref{SRTD} becomes asymptotically tight in the state-averaged sense. To this end, it shows the finite-size scaling of \(\langle D^{A}_{\alpha}(x)-D^{A}_{\alpha,\mathrm{conf}}(x)\rangle_{\mathcal{W}}\). Numerically, this difference decays exponentially with system size and appears to approach zero in the thermodynamic limit. This behavior suggests that, once the microcanonical average of the probability trace distance is suppressed, the upper bound in Proposition 1 effectively approaches an equality after state averaging. Under this interpretation, \(\langle D^{A}_{\alpha,\mathrm{prob}}\rangle_{\mathcal{W}}\) decays exponentially with system size, and the state-averaged trace distance is ultimately dominated by \(\langle D^{A}_{\alpha,\mathrm{conf}}\rangle_{\mathcal{W}}\).

\begin{figure}[!t]
    \centering
    \includegraphics[width=0.74\linewidth]{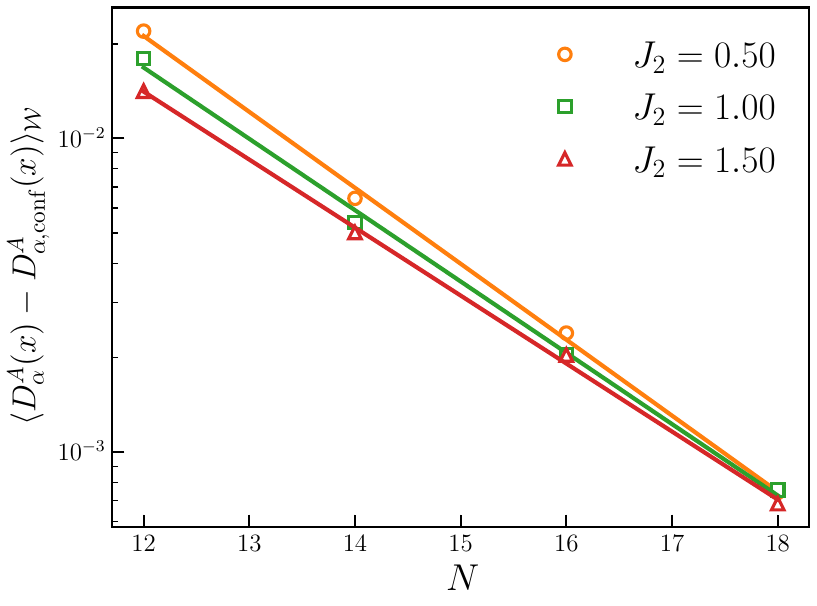}
    \caption{Finite-size scaling of \(\langle D^{A}_{\alpha}(x)-D^{A}_{\alpha,\mathrm{conf}}(x)\rangle_{\mathcal{W}}\). Different colors correspond to different \(J_2\) values.}
    \label{fig2}
\end{figure}

\textit{Conclusion.}---
We introduced a symmetry-resolved trace distance for SU(2)-symmetric many-body systems and showed that it separates into a probability trace distance and a configurational trace distance. The microcanonical average of the probability trace distance is bounded by fluctuations of the spin-sector probabilities within a microcanonical energy window, while in non-Abelian thermalizing systems the non-Abelian ETH describes the scaling of these fluctuations and hence the resulting exponential suppression with system size. Exact diagonalization of the \(J_1\)--\(J_2\) Heisenberg chain supports this picture: the fluctuations of the spin-sector probabilities decrease rapidly with system size, while the remaining finite-size behavior suggests that the total trace distance becomes asymptotically dominated by the configurational trace distance. More broadly, the block structure of the reduced density matrix separates the ETH-controlled fluctuations of the spin-sector probabilities from finer configurational fluctuations within each sector. This points to a useful strategy for diagnosing thermalization in systems with non-Abelian symmetries: rather than averaging over the symmetry structure, one should resolve it and ask which pieces are controlled by the non-Abelian ETH and which retain genuinely many-body information.

\begin{acknowledgments}
We are grateful to Lihui Pan, Hang Su and Wei Su for fruitful discussions. H.~S. and X.~W. were supported by MOST2022YFA1402701 and the NSFC Grant No.~11974244. J.~C. acknowledges the support of the Start-up Fund from Xihua University.
\end{acknowledgments}

\end{document}